\documentclass[twocolumn]{IEEEtran}

\usepackage{cite}
\usepackage{amsmath,amssymb,amsfonts}
\usepackage{bm}
\usepackage{amsthm}
\usepackage[linesnumbered,ruled,vlined]{algorithm2e}
\usepackage{graphicx}
\usepackage{textcomp}
\usepackage{xcolor}
\usepackage{comment}
\usepackage{svg}

\usepackage[margin=15mm,top=18 mm,bottom=26mm]{geometry}

\usepackage[normalem]{ulem} 
\usepackage{cancel}
\usepackage{mathtools, cuted}

\usepackage{tikz}
\usetikzlibrary{shapes,arrows,positioning}
\usetikzlibrary{3d}

\newcommand{\com}[1]{\textcolor{red}{#1}}
\newcommand{\comm}[1]{\textcolor{blue}{#1}}

\newtheorem{Lemma}{Lemma}

  {\proof}{\proofend}

\newcommand{\qn}{{\bf n}}

\newcommand{\qw}{{\bf w}}

\newcommand{\qI}{{\bf I}}

\newcommand{\qN}{{\bf N}}

\newcommand{\qY}{{\bf Y}}

\newcommand{\KE}{\mathcal{K}}
\newcommand{\LA}{\mathcal{L}}
\newcommand{\MRT}{\mathrm{MRT}}

\newcommand{\EEkm}{{\(E^E_{k} \,\)}} 

\newcommand{\gkl}{{\textbf{g}_{kl}}} 
\newcommand{\gklt}{{\textbf{g}^T_{kl}}} 
\newcommand{\gklc}{{\textbf{g}^*_{kl}}} 
\newcommand{\gklcp}{{\textbf{g}^*_{kl'}}} 

\newcommand{\gamkl}{{{\gamma}_{kl}}}
\newcommand{\gamil}{{{\gamma}_{il}}}
\newcommand{\gamklp}{{{\gamma}_{kl'}}}

\newcommand{\gamklsq}{{{\gamma}_{kl}^2}}

\newcommand{\aik}{{\alpha_{ik}}}
\newcommand{\aiksq}{{\alpha_{ik}^2}}
\newcommand{\aiktt}{{\alpha_{ik}^4}}
\newcommand{\oikl}{{{\varrho}_{ik,l}}} 
\newcommand{\oiklp}{{{\varrho}_{ik,l'}}}

\newcommand{\oiklsq}{{{\varrho}^2_{ik,l}}} 
\newcommand{\oiklsqp}{{{\varrho}^2_{ik,l'}}}

\newcommand{\kkl}{{K_{kl}}}

\newcommand{\kmrt}{{\kappa^{\MRT}_{kl}}}
\newcommand{\kmrti}{{\kappa^{\MRT}_{il}}}
\newcommand{\kmrtip}{{\kappa^{\MRT}_{il'}}}

\newcommand{\ckl}{{c_{kl}}}

\newcommand{\bkil}{{\zeta_{ik,l}}}
\newcommand{\bkilsq}{{\zeta^2_{ik,l}}}
\newcommand{\bkilp}{{\zeta_{ik,l'}}}
\newcommand{\bkilpsq}{{\zeta^2_{ik,l'}}}

\newcommand{\bklqp}{{\varsigma_{kl'}}}
\newcommand{\bilqp}{{\varsigma_{il'}}}

\newcommand{\bklq}{{\varsigma_{kl}}}
\newcommand{\bilq}{{\varsigma_{il}}}
\newcommand{\bklqsq}{{\varsigma^2_{kl}}}
\newcommand{\bilqsq}{{\varsigma^2_{il}}}

\newcommand{\gil}{{\textbf{g}_{il}}} 

\newcommand{\tauP}{{\tau_p\mathcal{P}_p}}

\newcommand{\hgkl}{{\hat{\textbf{g}}_{kl}}}

\newcommand{\thgkl}{{\Tilde{\hat{\textbf{g}}}_{kl}}}

\newcommand{\thgklc}{{\Tilde{\hat{\textbf{g}}}^*_{kl}}}
\newcommand{\thgklt}{{\Tilde{\hat{\textbf{g}}}^T_{kl}}}

\newcommand{\tekl}{{\Tilde{\hat{\pmb{\varepsilon}}}_{kl}}}
\newcommand{\teklt}{{\Tilde{\hat{\pmb{\varepsilon}}}^T_{kl}}}
\newcommand{\teklc}{{\Tilde{\hat{\pmb{\varepsilon}}}^*_{kl}}}

\newcommand{\thgil}{{\Tilde{\hat{\textbf{g}}}_{il}}}

\newcommand{\thgilc}{{\Tilde{\hat{\textbf{g}}}^*_{il}}}
\newcommand{\thgilt}{{\Tilde{\hat{\textbf{g}}}^T_{il}}}

\newcommand{\hgil}{{\hat{\textbf{g}}_{il}}}

\newcommand{\hgilt}{{\hat{\textbf{g}}^T_{il}}}
\newcommand{\hgilc}{{\hat{\textbf{g}}^*_{il}}}
\newcommand{\hgiltp}{{\hat{\textbf{g}}^T_{il'}}}

\newcommand{\bgkl}{{\Bar{\textbf{g}}_{kl}}} 
\newcommand{\bgklc}{{\Bar{\textbf{g}}^*_{kl}}} 
\newcommand{\bgklt}{{\Bar{\textbf{g}}^T_{kl}}} 

\newcommand{\bgilc}{{\Bar{\textbf{g}}^*_{il}}} 

\newcommand{\bhkl}
{{\Bar{\textbf{h}}_{kl}}} 

\newcommand{\bhilc}
{{\Bar{\textbf{h}}^*_{il}}} 

\newcommand{\bgilt}{{\Bar{\textbf{g}}^T_{il}}} 

\newcommand{\tgkl}{{\Tilde{\textbf{g}}_{kl}}} 

\newcommand{\tgil}{{\Tilde{\textbf{g}}_{il}}} 

\newcommand{\zekl}{{\xi_{kl}}} 
\newcommand{\bkl}{{\beta_{kl}}} 
\newcommand{\bil}{{\beta_{il}}} 
\newcommand{\bklsq}{{\beta^2_{kl}}} 
\newcommand{\bilsq}{{\beta^2_{il}}} 

\newcommand{\bklp}{{\beta_{kl'}}} 

\newcommand{\zkl}{{\textbf{z}_{kl}}} 
\newcommand{\sui}{{\sum\nolimits_{i \in \mathcal{P}_k}}} 
\newcommand{\suk}{{\sum\nolimits_{k\in\KE}}}
\newcommand{\suik}{{\sum\nolimits_{i\in\KE}}}
\newcommand{\sul}{{\sum\nolimits_{l\in\LA}}}

\newcommand{\slpL}{{\sum\nolimits_{l'\in\LA\setminus l}}}
\newcommand{\slp}{{\sum\nolimits_{l'\in\LA}}}

\newcommand{\xEl}{{\textbf{x}^E_{l}}} 

\newcommand{\rEk}{{r^E_{k}}} 
\newcommand{\IEk}{{{I}^E_{k}}} 
\newcommand{\EEk}{{{E}^E_{k}}} 
\newcommand{\ECk}{{{E}^C_{k}}} 
\newcommand{\dEk}{{\Delta {E}_k}} 

\newcommand{\ME}{{\mathbb{E}}} 
\newcommand{\MV}{{\mathrm{Var}}} 

\def\BibTeX{{\rm B\kern-.05em{\sc i\kern-.025em b}\kern-.08em
    T\kern-.1667em\lower.7ex\hbox{E}\kern-.125emX}}
    
\begin{document}

\pagenumbering{gobble}

\title{Energy Harvesting Characterization in Cell-Free Massive MIMO Using Markov Chains\\
}

\author{\IEEEauthorblockN{Muhammad Zeeshan Mumtaz, Mohammadali Mohammadi, Hien Quoc Ngo, and Michail Matthaiou}\\
\IEEEauthorblockA{{Centre for Wireless Innovation (CWI), Queen’s University Belfast, U.K.} \\
Email:\{mmumtaz01, m.mohammadi, hien.ngo, m.matthaiou\}@qub.ac.uk}
}

\maketitle

\begin{abstract}
This paper explores a discrete energy state transition model for energy harvesting (EH) in cell-free massive multiple-input multiple-output (CF-mMIMO) networks. Multiple-antenna access points (APs) provide wireless power and information to single-antenna UE equipment (UEs). The harvested energy at the UEs is used for both uplink (UL) training and data transmission. We investigate the energy transition probabilities based on the energy differential achieved in each coherence interval. A Markov chain-based stochastic process is introduced to characterize the evolving UE energy status. A detailed statistical model is developed for a non-linear EH circuit at the UEs, using the derived closed-form expressions for the mean and variance of the harvested energy. More specifically, simulation results confirm that the proposed Gamma distribution approximation can accurately capture the statistical behavior of the harvested energy. Furthermore, the energy state transitions are evaluated using the proposed Markov chain-based framework, while mathematical expressions for the self, positive and negative transition probabilities of the discrete energy states are also presented. Our numerical results depict that increasing the number of APs with a constant number of service antennas provides significant improvement in the positive energy state transition and reduces the negative transition probabilities of the overall network.
\let\thefootnote\relax\footnotetext{This work was
supported by the European Research Council
(ERC) under the European Union’s Horizon 2020 research
and innovation programme (grant agreement No. 101001331). The work of H. Q. Ngo was supported by the U.K. Research and Innovation Future Leaders Fellowships under Grant MR/X010635/1. M. Z. Mumtaz is also with the College of Aeronautical Engineering, National University of Sciences \& Technology (NUST), Pakistan, (email: zmumtaz@cae.nust.edu.pk).}
\end{abstract}


\vspace{-0.5em}
\section{Introduction}

CF-mMIMO emerges as a promising solution to the practical implementation of wireless power transfer (WPT). In CF-mMIMO, the distance between service provider APs and UEs is significantly decreased compared to conventional cellular architectures~\cite{Hien:cellfree}. This reduction in the AP-UE proximity mitigates the path loss attenuation effect in remote charging \cite{loku}. Moreover, continuous advancements in the massive MIMO hardware enable more precise electromagnetic beam alignment, along with side-lobe reduction, ensuring enhanced energy delivery at the UE's receive antennas.

The concept of EH in CF-mMIMO systems has been investigated in the recent literature~\cite{femenias,Demir,Wang:JIOT:2020, Mohammad:GC:2023}. In~\cite{femenias}, the UL information and DL power transfer in CF-mMIMO systems was considered, where a coupled UL/downlink (DL) optimization was performed through the UL signal-to-interference-noise ratio (SINR). Demir \textit{et al.} \cite{Demir} proposed a joint optimization of power control via the large-scale fading decoding vectors for maximizing the minimum spectral efficiency of all UEs. Wang \textit{et al.} \cite{Wang:JIOT:2020} conceived a completely wireless powered Internet of Things (IoT) system underpinned by a CF-mMIMO configuration, wherein the UEs rely on harvested energy during the UL data transmission. 
In~\cite{Mohammad:GC:2023}, a joint AP operation mode selection and power control design was developed. In this design, specific APs are assigned for energy transmission to energy UEs, while other APs are dedicated to information transmission towards information UEs.



Although the above mentioned research works~\cite{femenias,Demir,Wang:JIOT:2020,Zhang:IoT:2022, Mohammad:GC:2023} have shown the efficacy of CF-mMIMO in WPT applications, the device energy state and its transition over multiple coherence intervals, subject to the energy differential, has been overly neglected. These parameters are fundamentally important in determining the optimal EH strategy while utilizing wireless networks. In \cite{kusal}, the authors  considered the concept of Markov chain-based energy state transitions. However, they limited their scope of study to  the probability of fully charged state only. Moreover, the energy consumption has not been accounted for in existing Markov processes, which limits their applicability to only forward state transitions. Furthermore, harvest-store-use strategies have been presented  in \cite{fangchao,teng} for energy harvesting in wireless networks utilizing Markov process. However, these works have not rigorously characterized the statistics of harvested power, while focusing their analysis mainly on energy storage state transitions. Against this background, our research work focuses on the concept of a CF-mMIMO based EH mechanism using a Markov state transition process. The main contributions of our paper can be outlined as follows:

\begin{itemize}
    \item We propose a Markov chain-based state transition of discrete energy levels for the UE energy storage, which relies on EH from the APs. Moreover, we present the state transition probabilities in relation with the energy differential after power consumption during the UL training and data transmission phases.
    \item A new analytical solution is derived to obtain closed-form expressions for the statistical parameters (mean and variance) of non-linear EH, which are used to approximate the energy accumulation in a UE's battery with a Gamma distribution. Numerical results validate that this approximation is quite accurate and useful for further analysis of the transition probabilities.
    \item 
    Numerical results show that increasing the number of APs for a constant number of service antennas, improves significantly  the positive energy state transition probability. On the other hand, the probability of negative state transitions reduces with the same increase in the number of APs.
\end{itemize}

\textit{Notation:} We use bold upper case letters to denote matrices, and bold lower case letters to denote vectors. The superscripts $(\cdot)^T$ and $(\cdot)^*$  stand for the transpose and conjugate, respectively; $\ME\{\cdot\}$ and $\MV\{\cdot\}$ denote the statistical expectation and variance;
$\mathcal{N}_C(0,\sigma^2)$ denotes a circularly symmetric complex Gaussian random variable with variance $\sigma^2$. Moreover, $\delta(x,y)$ is the delta function that equals one if $x=y$ and equals zero otherwise; \(\Gamma(a)\) is the Gamma function, while \(\gamma (a, x )\) denotes the lower incomplete Gamma function~\cite[Eq. (8.350)]{Integral:Series:Ryzhik:1992}. 


\vspace{-1.2em}
\section{System Model}~\label{sec:sysmodel}
We consider a CF-mMIMO system under time division duplex operation, where \(L\) multiple-antenna APs provide both information and power services to \(K\) single-antenna UEs. For notational simplicity, we define the sets $\KE\triangleq \{1,\ldots,K\}$ and $\LA\triangleq \{1,\ldots,L\}$ as the collections of indices of the EUs and APs. Each AP is equipped with \(N \geq 1\) antennas, while each UE is equipped with one single antenna.  
We assume a quasi-static channel model, with each channel coherence interval spanning a duration of $\tau_c$ symbols. Each coherence block includes four phases: 1) UL training of duration $\tau_p$ symbols, 2) DL EH of duration $\tau_h$ symbols, 3) DL information transfer of duration $\tau_d$ symbols, and 4) UL information transfer of duration $\tau_u$ symbols.

Let \(\textbf{g}_{kl} \in \mathbb{C}^N\) denote the communication channel between the $k$-th UE and the $l$-th AP. 
We consider spatially uncorrelated Rician fading channels. Therefore, $\gkl$ can be expressed as 
\vspace{0.2em}
\begin{equation}~\label{eq:gkl}
    \gkl= \sqrt{\dfrac{\zekl}{\kkl+1}}
    \left(\sqrt{\kkl}\bhkl +\tgkl\right),
\end{equation}
where $\zekl$ denotes the large–scale fading coefficient, \(K_{kl}\) is the Rician \(K\)-factor, while
$\bhkl$ and $\tgkl$ correspond to the line-of-sight (LoS) and non-line-of-sight (NLoS) components, respectively. We assume that $\tgkl\sim\mathcal{N}_C(\boldsymbol{0}, \qI_N)$ and $\bhkl = [1,e^{j\pi\sin(\phi_{kl})},\ldots,e^{j(N-1)\pi\sin(\phi_{kl})} ]^T$. For  notational simplicity, we define $\beta_{kl}\! \triangleq\!{\dfrac{\zekl}{K_{kl}+1}}$, $\bklq\!\triangleq\! \bkl \kkl$, and
$\bgkl\triangleq \sqrt{\bklq} \bhkl$.  Therefore, the channel model in~\eqref{eq:gkl} can be rewritten as  
\vspace{0.2em}
\begin{equation}
    \gkl= \bgkl +\sqrt{\beta_{kl}}\tgkl.
\end{equation}

Similar to~\cite{Hien:Asilomar:2018}, we consider scenarios where the UEs are fixed or move slowly, but there are movements of objects around them. In these
scenarios, it is reasonable to assume that $\bgkl$ changes slowly with time, and is known a priori. In addition, $\beta_{kl}$ is also assumed to be known a priori.


\subsubsection{Uplink Training and Channel Estimation}
The first part of the coherence interval of length $\tau_p$ symbols will be used for the UL training phase to estimate the channels. During the training phase, all UEs simultaneously transmit known pilot sequences to the APs. 
Let \(\pmb{\varphi}_{k} \in \mathbb{C}^{\tau_p }\) represent  the pilot sequence transmitted from the $k$-th UE from the available set of pilot sequences, where \(\lVert\pmb{\varphi}_{k}\rVert^2= \tau_p\). We consider the practical case of $K>\tau_p$, which implies that different UEs may be sharing the same training sequence and, hence, \textit{pilot contamination} occurs~\cite{Hien:cellfree}. 
The received signal \(\qY_{p,l} \in \mathbb{C}^{N \times \tau_p}\) in the UL training phase at the $l$th AP is
$ \qY_{p,l}=\sum\nolimits_{k\in\KE} \sqrt{\mathcal{P}_p} \textbf{g}_{kl} \pmb{\varphi}^T_{k}+\qN_l,$
where \(\mathcal{P}_p\) is the pilot transmit power and \(\textbf{N}_l\) is the additive
white Gaussian noise (AWGN) matrix with independent, identically distributed \(\mathcal{N}_C(0,\sigma^2)\) entries. To obtain the estimate of $\gkl$, the projection of $\qY_{p,l}$ onto  $\pmb{\varphi}_k$ is first obtained as 
\vspace{0em}
\begin{equation}
  \zkl=\dfrac{\qY_{p,l}\pmb{\varphi}^*_{k}}{\sqrt{\tau_p}} = \sqrt{\tau_p\mathcal{P}_p}  \sui \gil+\qn_{kl},
\label{eq:equalized}
\end{equation}
where $\mathcal{P}_k\subset\{1,\ldots,K\}$ is the set of indices, including $k$, of UEs assigned with the same pilot as UE $k$ and $\qn_{kl}\triangleq \dfrac{\qN_l\pmb{\varphi}^*_{k}}{\sqrt{\tau_p}}\sim\mathcal{N}_C(\boldsymbol{0},\sigma^2\qI_N)$. Now, given $\zkl$, the minimum-mean-squared error (MMSE) estimate of $\gkl$ is $\hgkl=\bgkl+\thgkl$,
where $\thgkl\triangleq\ckl \big(\sqrt{\tau_p\mathcal{P}_p} \sui \sqrt{\bil} \tgil +\textbf{n}_{kl}\big)$ is the zero mean part of the  MMSE estimate with \(\ckl =\frac{\sqrt{\tau_p\mathcal{P}_p} \bkl}{ {\tau_p\mathcal{P}_p}\sui\bil +\sigma^2}\). The statistical parameters of the MMSE estimate are
\vspace{0.1em}
\begin{subequations}
   \begin{align}
    \ME\{\hgkl\}&=\bgkl,\\
    \MV\{ [\hgkl]_{t}\}&\triangleq \gamkl= \sqrt{\tauP} \bkl \ckl.
\end{align} 
\end{subequations}

We notice that when $i\in \mathcal{P}_k$,  the NLoS components of the MMSE estimates are linearly correlated as \(\thgil\!=\! \aik \thgkl\) where \(\aik\!\!\triangleq\!\! \bilsq/\bklsq\). Moreover, we can express the estimation error $\tekl$ as $\tekl\!\!=\!\!  \sqrt{\bkl} \tgkl\!- \thgkl$, with
$\tekl\!\sim\!\!\mathcal{N}_C(\boldsymbol{0},\!\upsilon_{kl}\qI_N)$ for $\upsilon_{kl}\!=\!\bkl\!-\gamkl$.

\vspace{-1.2em}
\section{Energy Harvesting and Markov Model}
In the DL EH phase, each AP forms precoding vectors using the estimated channels and transmits energy to the UEs. Let \(\qw_{kl} \in \mathbb{C}^{N}\)  denote the DL precoding vector for the EH phase. Maximum ratio transmission (MRT) has been shown to be an optimal beamformer for power transfer that maximizes the harvested energy when $N$ is large \cite{almradi2016performance}. Thus, the precoding vector at AP $l$ for UE $k$ can be designed as  $\qw_{kl}= \kmrt\hgkl$, where $\kmrt\!=\!\frac{1}{\sqrt{ \ME \left\{\lVert \hgkl \rVert^2 \right\}}}\!=\!\frac{1}{\sqrt{N(\bklq + \gamkl)}}$. Accordingly, the transmitted signal from the $l$th AP can be expressed as
\vspace{0em}
\begin{align}~\label{eq:xle}
    \xEl&=\suk\sqrt{\eta_{kl}}\qw^*_{kl}{e}_{k}, 
\end{align}
where ${e}_{k}$ is the zero-mean unit-variance energy signal for UE $k$.  We assume that independent energy symbols are used for different UEs. In~\eqref{eq:xle}, \(\eta_{kl}\) is the power control coefficient, chosen to satisfy the power constraint at each AP, i.e., $ \mathbb{E}\left\{\lVert \textbf{x}^E_l \rVert^2 \right\}\leq \mathcal{P}_d$,
where $\mathcal{P}_d$ is the maximum average transmit
power at AP $l$. 
The received signal at the $k$-th UE, during the EH phase, is given by
\begin{align}~\label{eq:rE}
    \rEk
    &=\sul \suik \kmrti \sqrt{\eta_{il}} \gklt \hgil^*{e}_{i}+n^E_k,
\end{align}
where \(n^E_k \sim \mathcal{N}_{\mathbb{C}}(0,\sigma^2)\) is the AWGN at the $k$-th UE. Since the noise floor is too low for EH, we neglect the effect of $n^E_k$ during the EH phase as in \cite{femenias,Demir,Wang:JIOT:2020,Boshkovska:CLET:2015}. Thus,  the received radio-frequency (RF) energy at EU $k$, denoted by $\IEk$, can be expressed as 
\begin{align}
\label{eq:harvested_power}
   \IEk
    &= \ME \bigg\{\Big\lvert \sul \suik \kmrti\sqrt{\eta_{il}}  \gklt \hgil^*  {e}_{i} \Big\rvert^2 \bigg\}
    \\
    &=  \!\sul\! \sum\nolimits_{l'\in \mathcal{L}}\! \suik \kmrti \kmrtip
    \sqrt{\eta_{il} \eta_{il'}} 
     \gklt \hgil^* \hgiltp \gklcp,\nonumber
    \end{align}
where the expectation is taken over $e_i$.

To characterize the harvested energy, 
a practical nonlinear EH model is considered \cite{Boshkovska:CLET:2015}. Therefore, the energy harvested within a single EH interval \(\tau_h\) at $k$-th UE is given by
\vspace{0.2em}
\begin{equation}
    \EEk= \tau_h \psi_k\left(\Lambda_k(\IEk) - \varphi_k\right),
\label{eq:harvested_energy}
\end{equation}
where  \(\psi_k= I^E_{k,max}/(1-\varphi_k)\)\ amd \(I^E_{k,max}\) is the maximum output DC power at UE $k$, $\varphi_k=1/{(1+e^{a_k b_k})}$ is a constant to guarantee a zero input/output response, while  $\Lambda\left(\IEk\right)=1/{(1+e^{-a_k(\IEk-b_k)})}$
is the traditional logistic function of \(\IEk\), where $a_k$ and $ b_k$ are circuit related parameters. We now present the statistics of $\IEk$, which will be used to develop the Markov chain model.  
\begin{Lemma}~\label{Lemma:EIK}
The statistics of the received RF energy, $\IEk$, can be obtained as
\begin{align}
    &\!\!\ME\big\{\IEk\big\}\!
    = \!\sul \slp \suik \!\kmrti\!\kmrtip\!\sqrt{\eta_{il}\eta_{il'}}
    \Xi_{ik,ll'},
    \label{eq:mean:final}
    \\
     &\!\! \MV\big\{\IEk\big\}
    \!=\! \sul \slp \suik \left(\kmrti \kmrtip\right)^2 {\eta_{il}\eta_{il'}}\nonumber\\
    & \quad\quad\quad\quad\,\,\times \big(\Upsilon_{ik,l}^{\mathsf{coh}}\delta_{l,l'} + \Upsilon_{ik,ll'}^{\mathsf{noncoh}}(1-\delta_{l,l'} )\big),
    \label{eq:var:final}
\end{align}
where
\vspace{-0.2em}
\begin{align}
    \Xi_{ik,ll'}&\!=\! \delta_{l,l'}N\big(N \bkilsq \oiklsq \!+\! \bklq \gamil \!+\!\bkl\nu_{il}\!+\!(1\!-\!\delta_{l,l'} )N^2\nonumber\\
    &\quad\times(\bkil \oikl\!+\! \aik \gamkl)(\bkilp \oiklp\!+\!\aik \gamklp)\nonumber\\
    &\quad+ \gamkl (\aiksq (N+1) \gamkl + 2 N  \aik \bkil \oikl -\gamil)\big)
\end{align}
with $\oikl \!=\! \bhkl^T \bhilc/N$, $\!\nu_{il}\!=\!\bilq\!+\!\gamil$ and $\bkil\!=\! \sqrt{\bklq \bilq}$. Moreover, $\Upsilon_{ik,l}^{\mathsf{coh}}$ and $\Upsilon_{ik,ll'}^{\mathsf{noncoh}}$ are given at~\eqref{eq:upsilon} at the top of next page. 
\begin{figure*}
  \begin{align}~\label{eq:upsilon}
        &\!\!\!\Upsilon_{ik,l}^{\mathsf{coh}} 
        \!=\!2 N^2 \bkilsq\! 
         \big[
         N \aiksq \bklq \oiklsq \gamkl\!\!+\!\!N \bkl \oiklsq \bilq\!\!+\!\! \aiksq \oiklsq \gamkl ( \bkl\!\!+\!\!N\gamkl\!)\!\!+\!\!\aiksq \bkl  \gamkl\! 
         \big]\!\! +\!\! N^2\!\big[\aiktt \bklqsq \gamklsq\!\!+\!\!\bklsq \bilqsq \big]\!\! +\! 2 \aiksq \!\gamkl N\!(N\!\!+\!\!1\!)\!  
           \nonumber\\
          & \hspace{-0.65em}
           \times\!\!\big[\aiksq  \bklq   \gamkl\!  \big((N\!\!+\!\!1\!)\gamkl\!\! +\! \beta_{kl} \!\big)\!\!+\! \bilq \!\big((N\!\!-\!\!1\!) \bkl\gamkl\!\! +\! \bklsq\!\! +\! 2 \gamklsq\big)\!\big]\!\!+\!\! N \aiktt \gamklsq\! \big[\! (\!N\!\!+\!\! 1)(\!N\!\!+\!\!2)\gamkl \big((\!N\!\!+\!\!3) \gamkl \! \!+\!\!4 \upsilon_{kl}\! \big)\!\! +\!\! \upsilon^2_{kl}(2N\!\!+\!\!1\!)\!\!-\!\!(\bkl\!\!+\!\!N\!\gamkl\!)^2 \big]\!,
                      \nonumber\\
        &\!\!\!\Upsilon_{ik,ll'}^{\mathsf{noncoh}}\!=\!
     N^2 \! \oiklsqp\bkilpsq \!\big[\aiksq\gamkl
     (\bkl\!\! +\!\! N\nu_{kl} )
      \!\!+\!\!
     N\bkl \bilq\! \big]\!\!+\!\!
    N^2\! \oiklsq \bkilsq\!\big[ \aiksq \gamklp \!(\bklp\!\!+\!\!N \nu_{kl'})\!\!+\!\!
     N\bklp \bilqp\! \big]\!\!+\!\!N^2\!\aiksq\! \big[\bkl  \bilq  \gamklp (\bklp\!\!+\!\!\bklqp\!\!
    \nonumber\\
    &\hspace{-0.65em}+\!N\!\gamklp\!)
    +\!\! \bklp \bilqp \gamkl(\beta_{kl}\!\!+\!\!\bklq\!\!+\!\!N\!\gamkl\!)\!\big]\!\!
    +\!\!
    N^2\!\bkl \bklp \bilq \bilqp\!\!+\!\!N^2\!\aiktt  \gamkl \gamklp\! \big[\!(\bklq \!\!+\!\! \bkl\!)(\bklqp\!\!+\!\!\bklp\!)\!\!
    +\!\!N\!(\gamkl(\bklp\!\!+\!\! \bklqp\!)\!\!+\!\!\gamklp(\bkl\!\!+\!\!\bklq\!)) \!
        \big].
  \end{align}  
\hrulefill
\vspace{-2mm}
\end{figure*}

\end{Lemma}
\begin{proof}
    See Appendix~\ref{Proof:Lemma:EIK}.
\end{proof}

\vspace{-0.2cm}

By invoking~\eqref{eq:harvested_energy}, we calculate the expected value of the harvested energy $\EEk$, which can be expressed as
\vspace{0.2em}
\begin{equation}
    \ME\left\{\EEk\right\}= \tau_h \psi_k \left( \ME\left\{\Lambda\left[\IEk\right]\right\}-\varphi_k \right).
\label{eq:exp_eek}
\end{equation}

We notice that deriving $\ME\big\{\Lambda\left[\IEk\right]\big\}$ is extremely challenging. Given that \(0 \leq\IEk <b_k \), the logistic function is convex in its nature. Thus, we  use the following approximation
\vspace{0.0em}
\begin{align}
\label{eq:jansen_inequality}
         \ME\big\{\Lambda\left[\IEk\right]\big\}\approx \Lambda\left[\ME\left\{\IEk\right\}\right]=\dfrac{1}{1+e^{-a_k\left(\ME\left\{\IEk\right\}-b_k\right)}}.
\end{align}

Therefore, $\ME\left\{\EEk\right\}$ can be expressed as 
\begin{align}
\label{eq:jansen_inequality2}
        &\ME\left\{\EEk\right\}
         =\psi_k\left(\Lambda\left[\ME\left\{\IEk\right\}\right]-\varphi_k\right)+\acute{\epsilon}_k,
\end{align}
where \(\acute{\epsilon}_k\) is the approximation gap. Note that even a Taylor series approximation of (\ref{eq:exp_eek}) at \(\ME\{\IEk\}\) yields the same expression as in (\ref{eq:jansen_inequality}), since all the higher order terms become negligible.
Now, by using~\eqref{eq:harvested_energy}, the variance of the harvested energy is calculated as
\begin{equation}
        \MV\left\{\EEk\right\}\!=\!(\tau_h \psi_{k})^2 \left(\ME\left\{\Lambda\left[\IEk\right]^2\right\}\!-\!\left(\ME\left\{\Lambda\left[\IEk\right]\right\}\right)^2\right).
\label{eq:var_eek}
\end{equation}

We can use a Taylor series expansion of \(\Lambda\left[\IEk\right]\) at \(\ME\{\IEk\}\) for the function \(\ME\{\Lambda\left[\IEk\right]^2\}\) in \eqref{eq:var_eek}, while neglecting the terms which involve odd powers of \((\IEk-\ME\{\IEk\})\) since their expectation becomes zero. Here, we  also neglect the fourth moment of \(\left(\IEk-\ME\{\IEk\}\right)\) which is very small given that \(\IEk\) is close enough to \(\ME\{\IEk\}\) with high probability due to the law of large numbers. Hence, we have
\small
\begin{align}
        \ME\Big\{\!\Lambda\left[\IEk\right]^2\Big\}
        &\!\hspace{0em} \approx\! \MV\big\{\IEk\big\}\!\bigg(\bigg(\frac{\partial \Lambda\left[\ME\{\IEk\}\right]}{\partial \IEk }\bigg)^2\!+\!\Lambda\!\left[\ME\{\IEk\}\right]
        \nonumber\\
        &\hspace{1em}\times\frac{\partial^2 \Lambda\left[\ME\{\IEk\}\right]}{\partial (\IEk)^2 }\bigg)\!+ \!\left(\Lambda\left[\ME\{\IEk\}\right]\right)^2.
\label{eq:firstVarexp}
\end{align}
\normalsize
Substituting \eqref{eq:firstVarexp} into \eqref{eq:var_eek}, the variance of the harvested energy is given by
\vspace{0em}
\small
\begin{align}
    \!\MV\left\{\EEk\right\}\!=& 
    (\tau_h \psi_{k})^2 \MV\Big\{\IEk\Big\}
    \bigg( \frac{\partial^2 \Lambda\left[\ME\{\IEk\}\right]}{\partial (\IEk)^2 }\left[\ME\{\IEk\}\right]\nonumber\\
    &+\bigg(\frac{\partial \Lambda\left[\ME\{\IEk\}\right]}{\partial \IEk }\bigg)^2\bigg)+\acute{\chi}_k,
\label{eq:var_eek_final}
\end{align}
\normalsize
where \(\acute{\chi}_k\) is the approximation error, and
\vspace{0.5em}
\begin{align}
        \!\!\frac{\partial \Lambda\!\left[\ME\{\IEk\}\right]}{\partial \IEk }\! &=\! \dfrac{a_k e^{-a_k\left(\ME\left\{\IEk\right\}-b_k\right)}}{\big(1+e^{-a_k(\ME\left\{\IEk\right\}-b_k)}\big)^2}
        \\
        \!\!\frac{\partial^2 \Lambda\!\left[\ME\{\IEk\}\right]}{\partial (\IEk)^2 } &\!=\! \dfrac{a^2_k e^{-a_k\left(\ME\left\{\IEk\right\}-b_k\right)}\!\big(1\!-\!e^{-a_k(\ME\left\{\IEk\right\}-b_k)}\big)}{\big(1+e^{-a_k(\ME\left\{\IEk\right\}-b_k)}\big)^3}.
\end{align}

Now, we focus on the probability distribution of the harvested energy using the above derived mean and variance expressions in \eqref{eq:jansen_inequality} and \eqref{eq:var_eek_final}. Conventionally, the energy storage by using a single potential source follows  an exponential distribution ~\cite[Eq. (27.33)]{halliday}. In this case, multiple signals are acting as potential energy sources as shown in \eqref{eq:rE}. 
Therefore, the composite harvested energy $\EEk$ will approximately follow a Gamma distribution as an effective sum of independent exponential distributions\cite{kusal,tavana}.

Using the mean and variance of the harvested energy of the $k$-th UE, we can approximate the shape and scale parameters of the Gamma distribution as
\vspace{0.2em}
\begin{align}
    \begin{split}
        k_{k}=\dfrac{\left(\ME\left\{\EEk\right\}\right)^2}{\MV\left\{\EEk\right\}},\quad
        \theta_{k}=\dfrac{\MV\left\{\EEk\right\}}{\ME\left\{\EEk\right\}}.
    \end{split}
\end{align}
The cumulative distribution function (CDF) of the harvested energy can be approximately given by
\vspace{0em}
\begin{equation}
    F\left(\EEk;k,\theta\right) \approx \dfrac{1}{\Gamma(k_k)} \gamma \left(k_k, \dfrac{\EEk}{\theta_k} \right).
\label{eq:gam_cdf}
\end{equation}

\vspace{-0.5em}
\subsection{Markov Chain-Based Energy State Transitions}
In this section, we propose a Markov chain-based probabilistic state transition between multiple discrete levels of the UE energy storage. Generally speaking, a Markov process  provides an appropriate characterization of a random process where the future state depends only on the present status of a system. In this context, the process of EH in the proposed WPT-based CF-mMIMO system can be represented via a Markov chain, as the state of UE energy storage in the next coherence interval is only dependent on the current energy state, and the energy differential between the energy harvested by the APs and the energy consumed during the UL training and data transmission. An illustration of this process is shown in Fig.~\ref{fig:Markov_chain}, where each node represents a discrete state of the UE energy while the directional arrows show the likelihood of transition from one state to another state.  

\begin{figure}[t]
    \centering
    \vspace{-1em}
    \includegraphics[trim=0 0cm 0cm 0cm,clip,width=\columnwidth]{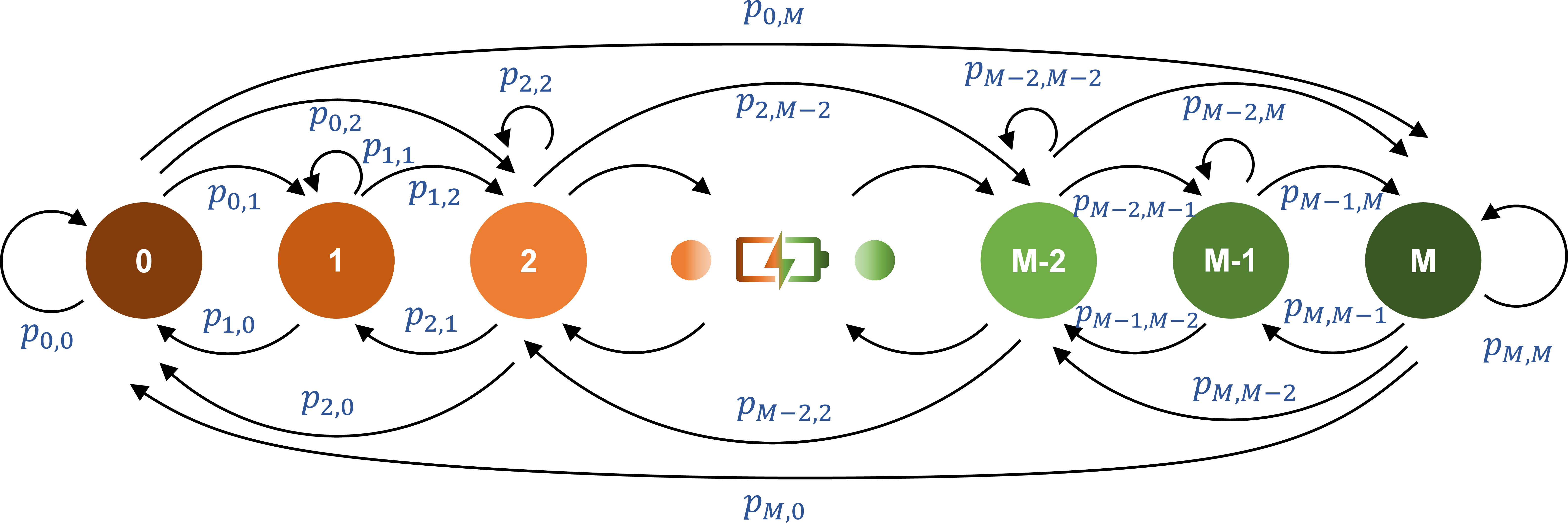}
     \vspace{-1.7em}
    \caption{Markov Chain-based energy state transitions.}
    \vspace{0.4em}
    \label{fig:Markov_chain}
\end{figure}
Each UE is assumed to have a finite energy storage with total capacity \(E_f\). The continuous energy entity is divided into \(M\) discrete energy levels with \(\delta E=E_f/M\) energy slots. Considering the independence of a Markov process from the past state history and its dependence only on the present state, it is equally likely that the UE energy level can be within any discrete state as we start observing the Markov process at any particular instance. So, the initial probability $\pi_0(m) \triangleq\Pr(X_0=m)$ of a particular energy state \(m\) at time \(t=0\) is \(1/M\). Moreover, we can represent the set of possible state transitions from a particular state \(i\) as \(\{P_{i,0},P_{i,1}, \hdots, P_{i,i-1}, P_{i,i}, P_{i,i+1}, \hdots, P_{i,M} \}\). Now, let us focus on this state \(i\) within the discrete nodes in Fig. \ref{fig:Markov_chain}, along with its adjacent states as shown in Fig. \ref{fig:State_transition_probabilities}.
If we consider that the UE energy level falls within this state at a certain time \(t\), the next possible states will depend on the energy differential \(\Delta E_k\) in the energy storage of the $k$-th UE within the single coherence interval \(\tau_c\). This factor \(\Delta E_k\) is further related to two entities: the total harvested energy, \EEkm, in the DL EH phase and the energy consumed both the UL pilot training  and UL information transfer phases, given by $\ECk\triangleq \tau_p P_p +\tau_u P_u$. Thus, \(\Delta E_k\) can be expressed as $\Delta E_k= \EEk -\ECk= \EEk-\left(\tau_p P_p +\tau_u P_u \right)$. To simplify the analysis, we consider that the expected energy differential is much smaller than the total state energy as \(\lvert\ME\{\Delta E_k\}\rvert\!\! <<\!\! E_f/M\), which will be validated in Section \ref{sec:results}. After this assumption, the set of possible transitions reduces to three states, i.e., \(\{P_{i,i-1}, P_{i,i}, P_{i,i+1}\}\). 



In Fig. \ref{fig:State_transition_probabilities}, \(P_{i,i+1}\) is the transition probability that a particular UE will have harvested enough energy in a particular \(\tau_c\) or that it will have positive energy differential \(\Delta E_k\) to transit from state \(i\) to a higher energy state \(i+1\), even after the power consumption in the UL pilot training and information transfer. Here, we assume that it is equally likely for all the energy states to be the initial energy state $i$ at the start of a Markov process, with probability $\pi^k_0(i)\!=\!\Pr(X^k_{0}\!=i )\!=\! 1/M \,\, \forall \,\, E^0_k \in \{E_1,E_2, \hdots, E_f\}$. At the $m$-{th} step of this Markov process, the total device energy will be the sum of initial battery energy $E^0_k$ and successive energy differentials $\Delta E^i_k$, given as, $E^n_k= E^0_k + \sum^n_{i=1} \Delta E^i_k$.

Considering the Markov property of independence across consecutive steps, the probability of the device energy at the $m${-th} interval being in state $j$, is only dependent on the last interval state $i$, irrespective of the energy state history in previous intervals. Therefore, the energy differential $\dEk$ during the $m$-{th} interval is the sole decision variable to determine the future state $X^K_m$, for a given previous state $X^K_{m-1}\!\!=\!i$. It can be inferred, by visualizing Fig. \ref{fig:State_transition_probabilities}, that if the current energy value in a particular state is within $\ME\{\dEk\}$ proximity to the state boundary, then transition to an adjacent energy state is highly probable. Hence, two statistical characteristics of $\dEk$ determine the conditional probability of state transition: $\ME\{\dEk\}$ and $\Pr(\dEk\!\!\leq\!0)$ or $\Pr(\dEk\!\!>\!0)$. It should be noted that $\ME\{\dEk\}$ is related to the magnitude of energy transition, whereas $\Pr(\dEk\!\!\leq\!0)$ and $\Pr(\dEk\!\!>\!0)$ specify the direction of state transition. 

\begin{figure}[t]
    \centering
    \vspace{-0.9em}
    \includegraphics[trim=0 0cm 0cm 0cm,clip,width=0.7\columnwidth]{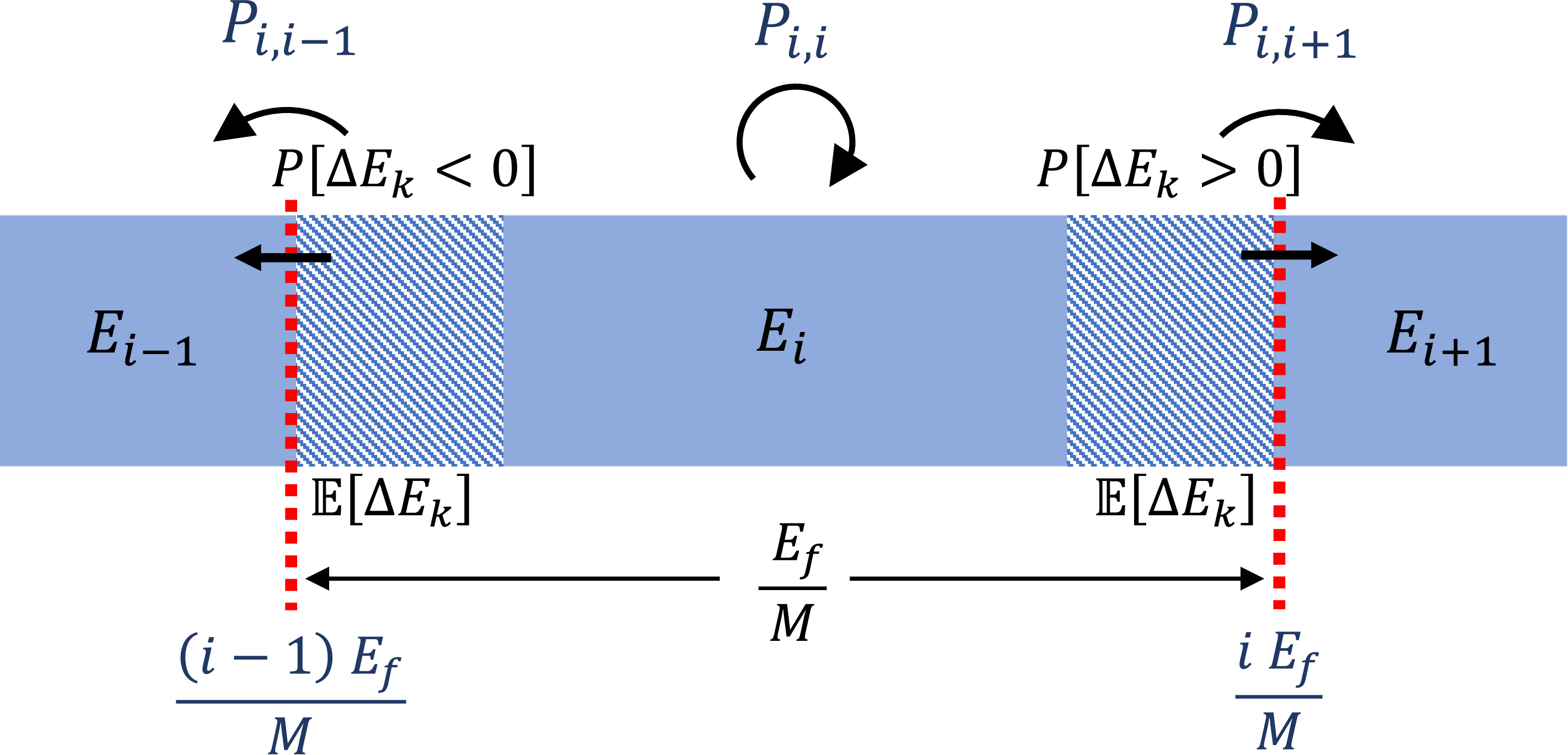}
    \vspace{-0.7em}
    \caption{Transition probability cases from a particular energy state.}
    \vspace{1em}
    \label{fig:State_transition_probabilities}
    \vspace{-0.5em}
\end{figure}
\vspace{-0.04em}

We first discuss the case of self transition within the same state over the next coherence interval. If the energy level in the $(m-1)$-{th} interval lies in between $\big((i-1)E_f/M + \dEk\big)$ and $\big((i+1)E_f/M - \dEk\big)$, then the UE energy storage will maintain state $i$. The probability of this event ($P_{i,i}$) is $(1- M \ME\{\Delta E_k\}/{E_f})$, which means that the energy level is not in the vicinity of the upper and lower boundaries of state $i$ with $\dEk$ margin on average. The probability of state transitions, $P_{i,i+1}$ and $P_{i,i-1}$, involve the condition of negation event ($\bar{P}_{i,i}$) of self transition. Given this condition, a negative state transition will happen whenever the energy consumption is more than the harvested energy (i.e., $\EEk\!\!\leq\!\!\ECk$). Using the CDF expression of the harvested energy in \eqref{eq:gam_cdf}, we can show that 
\vspace{0.3em}
\begin{align}
    \begin{split}
        \Pr\left(\Delta E_k \!\leq\! 0 \right)\!=\! \Pr\left(\EEk \leq \ECk \right)\!\approx\!\dfrac{1}{\Gamma(k_k)} \gamma \left(k_k, \dfrac{\ECk}{\theta_k} \right).\\
    \end{split}
\label{eq:neg_prob}
\end{align}

Similarly, the probability of positive state transition can be expressed as $\Pr\!\left(\Delta E_k \! >\! 0 \right)\!\!=\! \Pr\left(\EEk \!>\! \ECk \right)=1-\Pr\left(\Delta E_k \!\leq\! 0 \right)$.

Summarizing the above explanation, the conditional probability of transition from a particular state $i$ in the $(m-1)$-{th} interval to  state $j$ in the $m$-{th} interval for UE $k$ can be represented as

\begin{align}
\centering
    \begin{split}    
    \label{eq:transition_prob}
        P^k_{i,j}&= \Pr\left(X^k_{m}=j \Bigl\lvert X^k_{m-1}=i \right)\\
        &\hspace{-2em}=  \Pr\left(\!\dfrac{(j-1)E_f}{M}\le E^{m}_k\le\! \dfrac{j E_f}{M} \Bigl\lvert \dfrac{(i-1)E_f}{M}\le E^{m-1}_k\le \dfrac{i E_f}{M}\right)\\
        &\hspace{-2em}\approx
        \begin{cases}
            1- \dfrac{M \ME\{\Delta E_k\}}{E_f} & \text{for} \quad j=i \\
            \dfrac{M \ME\{\Delta E_k\}}{E_f \Gamma(k_k)}  \gamma \Big(k_k, \dfrac{\ECk}{\theta_k} \Big) & \text{for} \quad j=i+1 \\
            \dfrac{M \ME\{\Delta E_k\}}{E_f} \Bigg( 1-\dfrac{1}{\Gamma(k_k)} \gamma \left(k_k, \dfrac{\ECk}{\theta_k} \right)\Bigg) & \text{for} \quad j\!=\!i-1, 
         \end{cases}
    \end{split}
\end{align}
where $\ME\{\Delta E_k\}=\tau_h \psi_k \left( \ME\left\{\Lambda\left[\IEk\right]\right\}-\varphi_k \right)-\ECk$, while $\ME\left\{\Lambda\left[\IEk\right]\right\}$ is given in~\eqref{eq:jansen_inequality}.

Using the Markov property of state independence and definition of conditional transition probability, we can extend our analysis to gain valuable insights into an important performance indicator for EH systems: state transition probabilities after a period of \(n\) coherence intervals. This parameter can provide a useful understanding of the actual energy accumulation within a certain number of coherence intervals. Given that the initial state is \(q\), this quantity can be expressed as
\begin{equation}
    \pi^k_n(j)=\Pr(X^k_{n}=j )= \Bigg ( \prod\nolimits^n_{m=1} \sum\nolimits^M_{i=1} P^k_{i,j} \Bigg ) \pi^k_0(q).
\label{eq:trans_prob}
\end{equation}

\section{Numerical Results}
\label{sec:results}
In this section, we simulate the proposed CF-mMIMO model in Section \ref{sec:sysmodel} using Monte-Carlo simulations over 2000 coherence intervals for each configuration, where $L$ APs are providing energy along with data transmission to $K$ UEs, distributed over a square area of $100 \times 100$~m$^2$. Here, $K=20$ UEs have been considered in this area with battery storage capacity of $300$~mJ \cite{Wang:JIOT:2020}. This energy storage is divided into $M=2,000$ discrete states to observe the demonstration of Markov process. Moreover, we limit the total network power to 10W to make a fair comparison for varying network parameters. Similarly, the total number of services antenna, $LN$, is kept constant to $288$ elements \cite{Demir}.

The large scale fading parameter is defined as $\zekl=10^{-(PL_{kl}+\Psi_{kl})/10}$, which models the path loss $PL_{kl}$ and log-normal shadowing $\Psi_{kl}$. We adopt the three-slope model, which has been used to model path loss $PL_{kl}$ (dB) in~\cite{Hien:cellfree}, 
\begin{equation}
\hspace{-1em}PL_{kl}\! =\! 
\begin{cases} 
   -L\!-\!35 \,{\log}_{10}(d_{kl}) &  d_{kl}\!>\! d_1, \\
   -L\!-\!15 \,{\log}_{10}(d_{1})\!-\!20 \,{\log}_{10}(d_{kl}) &  d_0\!< \!d_{kl} \!< \!d_1, \\
   -L\!-\!15 \,{\log}_{10}(d_{1})\!-\!20 \,{\log}_{10}(d_{0}) &  d_{kl}\! <\! d_0,\nonumber \\
\end{cases}
\end{equation}
where $d_0=10$ m, $d_1=50$ m, and
$ L\triangleq \,46.3+33.9 \, \log_{10}(f)-13.82 \,\log_{10}(h_{AP})-(1.1 \log_{10} (f)-0.7) h_s +(1.56 \,\log_{10}(f)-0.8),$
with carrier frequency $f=1900$ MHz, AP height $h_{AP}= 15$ m, and UE height $h_s=1.65$ m, as given in \cite{Hien:cellfree, Wang:JIOT:2020}. Note that the system bandwidth has been fixed at $20$ MHz with coherence interval as $\tau_c=0.2$s, each with $200$ samples.  

\begin{figure}[t]
    \centering
\vspace{-2em}    
\includegraphics[trim=0 0cm 0cm 0cm,clip,width=0.75\columnwidth]{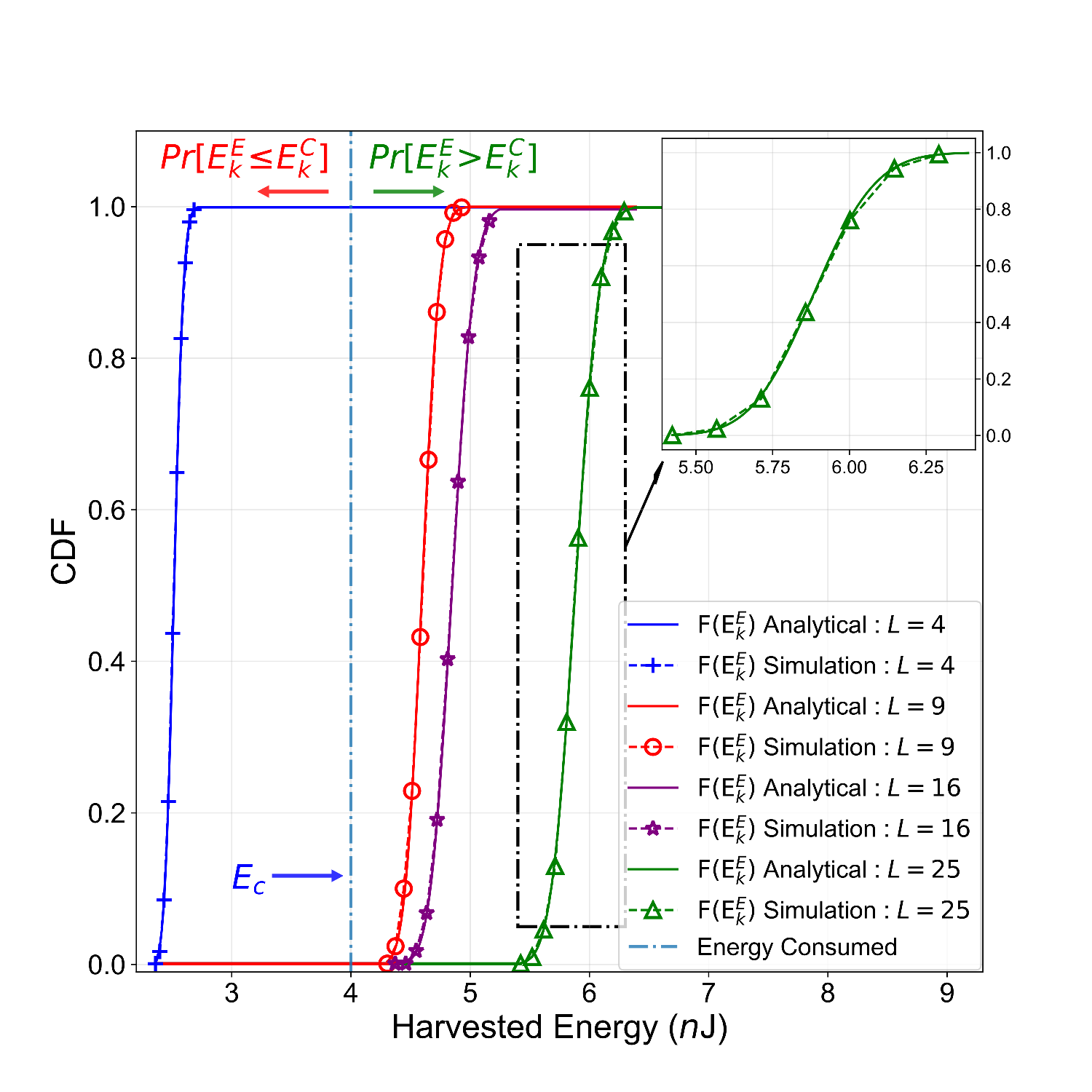}
    \vspace{-1.8em}
    \caption{CDF of the harvested energy per coherence interval for different number of APs with $LN=288$.}
    \vspace{0.2em}
    \label{fig:CDF}
\end{figure}

We first demonstrate the CDF of the harvested energy, i.e., $F\left(\EEk;k,\theta\right)$, of the user with median energy level among all UEs  for different AP configurations in Fig. \ref{fig:CDF}. The justification for this particular user choice is that we have to consider the probability distribution of the harvested energy of individual users for comparison with our analytical results. In this context, the CDF of median energy level user provides adequate representation for all users for a specific AP configuration. It can be noticed that the simulation results (dotted lines with markers) match the analytical results (solid lines) as given in \eqref{eq:gam_cdf} very closely, which validates our approximation that the probability distribution of the harvested energy can be modelled using a Gamma distribution. The analysis of the zoom plot for $L=25$ in Fig. \ref{fig:CDF} further supports our assumption, since it illustrates the tightness of the Gamma approximation based on the statistics of the harvested energy given in \eqref{eq:jansen_inequality2} and \eqref{eq:var_eek_final}. Similar results have been observed for all users within the network. Moreover, the graph in Fig. \ref{fig:CDF} has been divided into two regions around the decision boundary of energy consumption ($\ECk$) in a single coherence interval. The area which lies on the left of the decision boundary represents the region of negative energy differential ($\Pr(\EEk \leq \ECk)$), whereas the right side shows the region of positive energy differential ($\Pr(\EEk > \ECk)$). It is clear that increasing the number of APs in the CF-mMIMO system from $L=4$ to $L=25$, while keeping the number of service antennas constant, leads to a significant improvement in harvested energy. At $F\left(\EEk;k,\theta\right)=0.5$, the performance gains for cases $L=9, 16, 25$ in comparison to $L=4$ case, are 82.6\%, 92.7\% and 133.8\% respectively.  We can observe that for the case $L=4$, it is highly probable that the harvested energy will be less than the consumed energy, resulting in negative differentials for this instance. Conversely, positive energy differentials are likely for the cases $L=9, 16$ and $25$, affirming the superiority of the CF-mMIMO configuration over the co-located mMIMO in EH applications.

\begin{figure}[t]
\vspace{-2em}
    \centering
    \includegraphics[trim=0 0cm 0cm 0cm,clip,width=0.75\columnwidth]{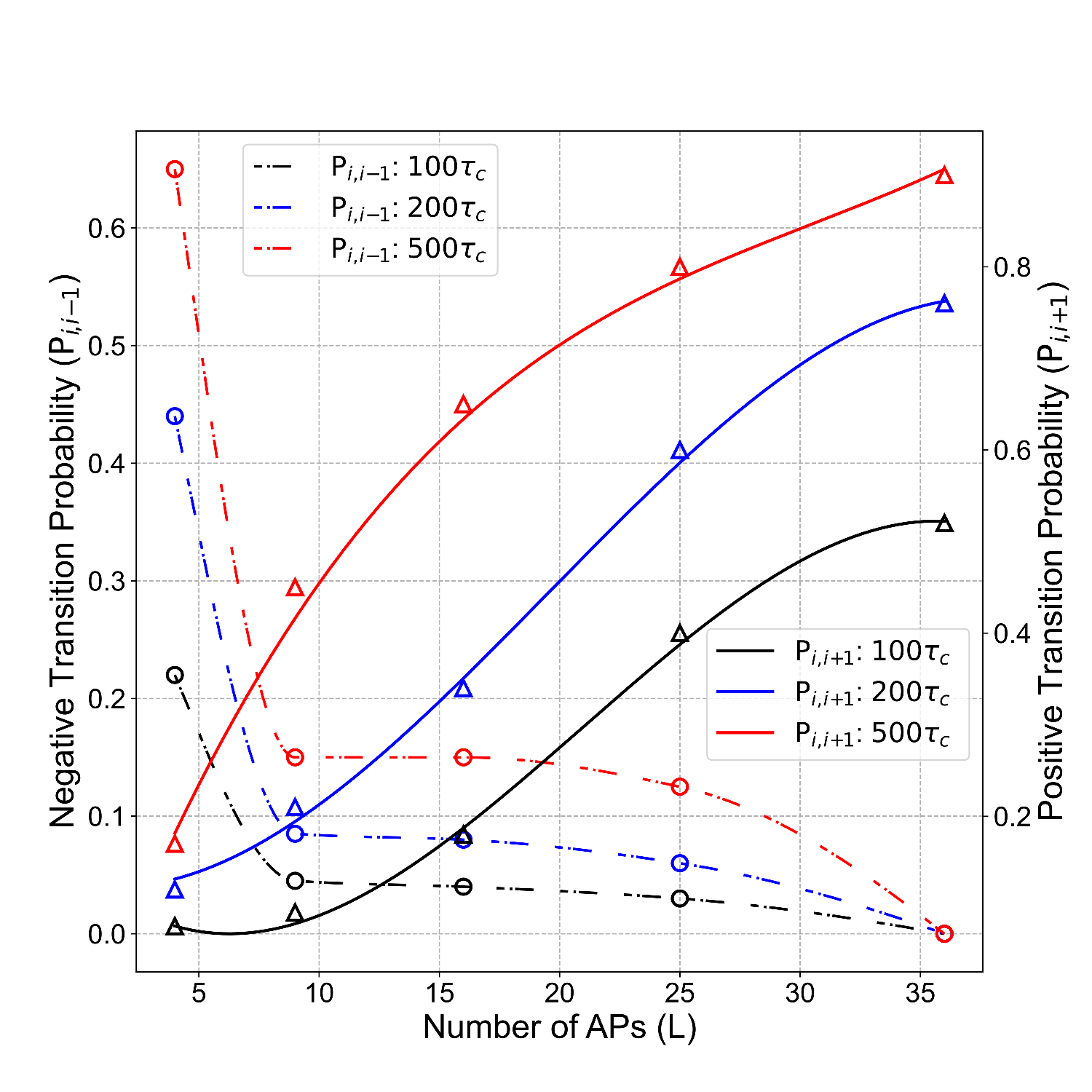}
    \vspace{-1.8em}
    \caption{Energy state transition probabilities after $n$ coherence intervals. Markers depict the simulation results and lines show the analytical results.}
    \vspace{0.2em}
    \label{fig:trans_prob}
\end{figure}

The transition probabilities of discrete user storage energy using the Markov chain model are tabulated in Table~\ref{tab:abbreviation}. Note that the harvested energy  within a single coherence interval ($\sim 10^{-6}-10^{-5}$$\mu$J) is insignificant in comparison to the energy of a discrete state ($0.15$~mJ). Therefore, the probabilities of higher order transitions ($P_{i,i\pm 2},P_{i,i\pm 3}, \hdots, P_{i,i\pm M}$) are almost zero. Figure~\ref{fig:trans_prob} shows the state transition of discrete user energy storage  over steps in order of coherence intervals using \eqref{eq:trans_prob}. 
We observe probability transitions of energy states over periods $n=100, 200 , 500\, \tau_c$,  It can be further noted that both the positive and negative transition probabilities are higher for larger $n$, as successive events of EH in a single interval have a cumulative effect, which is pronounced for larger observation intervals.  We can also observe a considerable increase in the positive transition probabilities and decrease in the negative transition probabilities with an increasing number of APs. Considering the discussion on Fig. \ref{fig:CDF}, it can be noticed that for the case $L=4$, the negative transition probabilities are more than the positive transition probabilities. This trend changes for the cases $L=9, 16$ and $25$, as the harvested energy is more than the consumed energy. It should be noted here that self-transition probability $P_{i,i}$$\!=\! 1\!-\!(P_{i,i-1}\!+\!P_{i,i+1})$ has not been shown here, which decreases with an increasing number of APs, due to an increase of the positive transition probabilities.  
We can attribute this improvement predominantly to more energy source diversification in CF-mMIMO, as smaller distances between APs and UEs results in a significant decrease in the path loss.


\begin{table}[t]
    \caption{Transition Probabilities of UE Energy States}
    \vspace{-1em}
    \centering
    \footnotesize
    \begin{tabular}{|p{1.3cm}|p{1cm}|p{1cm}|p{1cm}|}
        \hline
        No. of APs & $P_{i,i}$ & $P_{i,i+1}$ & $P_{i,i-1}$ \\
        \hline
        \textbf{4} & 0.99710 & 0.0007 & 0.00220\\
        \hline
        \textbf{9} & 0.99845 & 0.0011 &0.00045\\
        \hline
        \textbf{16} & 0.99800 & 0.0016 &0.00040\\
        \hline
        \textbf{25} & 0.99570 & 0.0040 &0.00030\\
        \hline
        \textbf{36} & 0.99480 & 0.0052 &0
        \\\hline
    \end{tabular}
    \label{tab:abbreviation}
\end{table} 






\section{Conclusion}
We formulated a new framework for energy harvesting in CF-mMIMO by capitalizing on the discrete energy state transitions using a Markov stochastic process. Using a Gamma distribution approximation, closed-form expressions for the statistical parameters of the harvested energy were derived alongside the state transition probabilities. We have further analyzed the impact of AP dispersion on the energy harvesting performance in terms of state transition probabilities. The future extension of this work would include power control optimization based on energy state transition probabilities to further enhance the energy harvesting efficiency of a realistic CF-mMIMO system. Henceforward, a detailed analysis of steady state probabilities can also provide useful insights into the importance of Markov processes for energy harvesting.


\appendices

\vspace{-0.1cm}
\section{Proof of Lemma~\ref{Lemma:EIK}}
\label{Proof:Lemma:EIK}
By taking the expectation of~\eqref{eq:harvested_power} over the small-scale fading, we get
\begin{align}
    \ME\big\{\IEk\big\}
    =&  \sul \slp\suik 
    \kmrti \kmrtip\sqrt{\eta_{il}\eta_{il'}}    
    \nonumber\\
    &\hspace{2em}\times
    \ME\big\{\gklt \hgilc \hgiltp \gklcp \big\}.
\label{eq:harvested_power_final}
\end{align}
Now, we consider two cases as follows:

\textbf{Case 1)} $l'= l$: The expectation in~\eqref{eq:harvested_power_final} can be obtained as
\begin{align}~\label{eq:mean:case1}
        &\ME \big\{\gklt \hgilc \hgilt \gklc\big\}
        =\bgklt\bgilc\bgilt\bgklc\!\!+\!\ME\{\bgklt\thgilc\thgilt\bgklc\}\nonumber\\
        &\hspace{3em}+\!\ME\{\thgklt\bgilc\bgilt\thgklc\}
        +\!\ME\{\thgklt\thgilc\thgilt\thgklc\}\!\!+\!\ME\{\teklt\bgilc\bgilt\teklc\}\nonumber\\
        &\hspace{3em}+\!\ME\{\teklt\thgilc\thgilt\teklc\}
        \!+\!\ME\{\bgklt\bgilc\thgilt\thgklc\}\!+\! \ME\{\thgklt\thgilc\bgilt\bgklc\}\nonumber\\
        &\hspace{3em}=N\big( N \bkilsq \oiklsq \!\!+\! \bklq \gamil+\! \bilq \gamkl\!\!+\! \aiksq (N+1) \gamklsq
        \nonumber\\
        &\hspace{3em} +\!(\bilq+\gamil) (\bkl\!-\!\gamkl)\!\!+\!  2 N \gamkl \aik \bkil \oikl \big),
\end{align}
where, we have used 
\begin{align}
        \ME\{\thgklt\thgilc\thgilt\thgklc\}\!\!=\! \aiksq \ME\{\thgklt\thgklc\thgklt\thgklc\}\!\!=\!\aiksq N(N\!+\!1) \gamklsq.
\end{align}

\textbf{Case 2)} $l'\neq l$: Following the independence of the channels, we get 
\begin{align}~\label{eq:mean:case2}
    &\ME\left\{\gklt \hgilc \hgiltp \gklcp \right\} = \ME\left\{ \gklt \hgilc\right\}\ME\big \{\hgiltp \gklcp\big\}\nonumber\\
    &\hspace{2em}= N^2(\bkil \oikl+ \aik \gamkl)(\bkilp \oiklp+\aik \gamklp).
\end{align}
To this end, by substituting~\eqref{eq:mean:case1} and~\eqref{eq:mean:case2} into~\eqref{eq:harvested_power_final}, the desired result in~\eqref{eq:mean:final} can be obtained.

Now, we focus on the calculation of the variance of the harvested power.  We can write
\begin{align}
    \MV\big\{\IEk\big\}
    &=  \sul \suik \left(\kmrti\right)^4 \eta^2_{il} \MV\big\{\gklt \hgil^* \hgilt \gklc \big\}\nonumber\\
    & + \sul \slpL \suik \left(\kmrti \kmrtip\right)^2 {\eta_{il}\eta_{il'}}\nonumber\\
    &\times\MV \big\{ \gklt \hgil^* \hgiltp \gklcp\big\}.
\label{eq:harvested_power_var_final}
\end{align}
The derivation of the two variance terms follows the same procedure as the expected value of $\IEk$, though it is omitted here due to space limitations.


\vspace{-.2cm}
\bibliographystyle{IEEEtran}
\bibliography{main}

\end{document}